\documentclass[12pt]{article}

\usepackage{amsmath,amssymb,amsthm,fullpage,charter,verbatim,calc}
\usepackage{eepic,epic}
\usepackage{vmargin}
\usepackage{dsfont}
\setpapersize{A4}
\setmarginsrb{30mm}{20mm}{30mm}{20mm}{12pt}{11mm}{0pt}{11mm}

\newtheorem{theorem}{Theorem}
\newtheorem{lemma}{Lemma}
\newtheorem{corollary}{Corollary}
\newtheorem{conjecture}{Conjecture}

\newcommand{\eps}{\epsilon}

\title{\bf Tight Approximation Bounds for Vertex Cover on Dense $k$-Partite Hypergraphs\\[1ex]}
\author{
Marek Karpinski\thanks{Dept. of Computer Science and the Hausdorff
    Center for Mathematics, University of Bonn.
    Supported in part by DFG grants and the Hausdorff Center grant EXC59-1.
    Email:~\texttt{marek@cs.uni-bonn.de} } \and
Richard Schmied\thanks{Dept. of Computer Science, University of Bonn.
    Work supported by Hausdorff Doctoral Fellowship.
    Email:~\texttt{schmied@cs.uni-bonn.de}} \and
Claus Viehmann\thanks{Dept. of Computer Science, University of Bonn.
    Work partially supported by Hausdorff Center for Mathematics, Bonn.
    Email:~\texttt{viehmann@cs.uni-bonn.de}}
}
\date{}

\begin{document}
\maketitle

\begin{abstract}
We establish almost tight upper and lower approximation bounds for
the Vertex Cover problem on dense $k$-partite hypergraphs.
\end{abstract}

\section{Introduction}
A hypergraph $H=(V,E)$ consists of a vertex set $V$ and a collection of hyperedges $E$
where a hyperedge is a subset of $V$. $H$ is called $k$-uniform if every edge in $E$
contains exactly $k$ vertices. A subset $C$ of $V$ is a vertex cover of $H$ if every
edge $e\in E$ contains at least a vertex of $C$.

The \emph{Vertex Cover} problem in a $k$-uniform hypergraph $H$ is the problem of computing
a minimum cardinality vertex cover in $H$.
It is well known that the problem is $NP$-hard even for $k=2$ (cf.~\cite{K75}).
On the other hand, the simple greedy heuristic which chooses a maximal set of nonintersecting edges, and then
outputs all vertices in those edges, gives a $k$-approximation algorithm for 
the Vertex Cover problem restricted to $k$-uniform hypergraphs.  
The best known approximation algorithm achieves a slightly better approximation ratio of $(1-o(1))k$ and is due to Halperin~\cite{H02}.

On the intractability side, Trevisan~\cite{T01} provided one of the first inapproximability results
 for the $k$-uniform vertex
cover problem and obtained a inapproximability factor of $k^{\frac 1{19}}$ assuming $P \neq NP$.
In 2002, Holmerin~\cite{H02} improved the factor to $k^{1-\eps}$. 
Dinur et al.~\cite{DGK02,DGKR05} gave consecutively two lower bounds, first $(k-3-\eps)$ and
later on $(k-1-\eps)$. Moreover, assuming Khot's Unique Games Conjecture (UGC)~\cite{K02}, Khot and Regev~\cite{KR08} proved an inapproximability factor of $k-\eps$ for the  Vertex Cover problem on $k$-uniform hypergraphs. Therefore, it implies that the currently achieved ratios are the best possible. 

The Vertex Cover problem restricted to $k$-partite $k$-uniform hypergraphs, when the underlying
partition is given, was studied by  Lov\'asz~\cite{L75} who achieved a $\frac k2$-approximation.
This approximation upper bound is obtained by rounding the natural LP relaxation of the problem.
The above bound on the integrality gap was shown to be tight in~\cite{AHK96}.
As for the lower bounds, Guruswami and Saket~\cite{GS10b} proved that it is NP-hard to approximate the Vertex Cover problem on
$k$-partite $k$-uniform hypergraphs to within a factor of $\frac k4-\eps$ for $k\geq 5$. 
Assuming the Unique Games Conjecture, they also provided an inapproximability factor of 
$\frac k2-\eps$ for $k\geq3$. More recently, Sachdeva and Saket~\cite{SS11} claimed a nearly optimal $NP$-hardness factor. 

To gain better insights on lower bounds, dense instances of many optimization problems has been intensively studied~\cite{AKK95,KRS09,KZ97,K01}.
The Vertex Cover problem has been investigated in the case of dense graphs, where the number of edges is within a constant factor of $n^2$, by Karpinski and Zelikovsky~\cite{KZ97}, Eremeev~\cite{E99}, Clementi and Trevisan~\cite{CT99}, later by Bar-Yehuda and Kehat~\cite{BK04} as well as Imamura and Iwama~\cite{II05}.

The Vertex Cover problem restricted to dense balanced $k$-partite $k$-uniform hypergraphs 
was introduced and studied in~\cite{CKSV10b}, where it was proved that this restricted version
of the problem 
admits an approximation ratio better than $\frac k2$ if the given hypergraph is dense enough.

In this paper, we give a new approximation algorithm for the  
Vertex Cover problem restricted to dense $k$-partite $k$-uniform hypergraphs
and prove
that the achieved approximation ratio is almost tight assuming the 
Unique Games Conjecture.

\section{Definitions and Notations}
Given a natural number $i\in \mathds{N}$, we introduce for notational simplicity the set $[i]=\{1,..,i\}$
and set $[0]=\emptyset$.
Let $S$ be a finite set with cardinality $s$ and $k\in [s]$. We will use the abbreviation ${S\choose k}=\{S'\subseteq S\mid |S'|=k \}$.\\
A \emph{$k$-uniform hypergraph} $H=(V(H),E(H))$ consists of a set of vertices $V$ and a collection $E\subseteq{V\choose k}$ of edges.
For a $k$-uniform hypergraph $H$ and a vertex $v\in V(H)$, we define the \emph{neighborhood} $N_H(v)$ of $v$ by
$\left(\bigcup_{e\in \{e\in E\mid v\in e\}} e\right)\setminus \{v\}$ and the
\emph{degree} $d_H(v)$ of $v$ to be $|\{e\in E\mid v\in e\}|$.
We extend this notion to subsets of $V(H)$, where $S\subseteq V(H)$ obtains the degree $d_H(S)$ by $|\{e\in E\mid S\subseteq e\}|$.
\\ 
A \emph{$k$-partite} $k$-uniform hypergraph $H=(V_1,..,V_k,E(H))$ is a $k$-uniform hypergraph 
such that $V$ is a disjoint union of 
$V_1,..,V_k$ with $|V_i\cap e|=1$ for every $e\in E$ and $i\in [k]$. In the remainder, we assume
that $|V_i|\geq |V_{i+1}|$ for all $i\in [k-1]$ and $k=O(1)$.\\
A \emph{balanced} $k$-partite $k$-uniform hypergraph $H=(V_1,..,V_k,E(H))$ is a 
$k$-partite $k$-uniform hypergraph with $|V_i|=\frac{|V|}{k}$ for all $i\in [k]$.
 We set $n=|V|$ and $m=|E|$ as usual.\\ 
For a $k$-partite $k$-uniform hypergraph 
$H=(V_1,..,V_k,E(H))$ and $v\in V_k$, we introduce the \emph{$v$-induced hypergraph} $H(v)$,
where the edge set of $H(v)$ is defined by $\{e\setminus \{v\} \mid v\in e\in E(H)\}$
and the vertex set of $H(v)$ is  partitioned into $V_i\cap N_H(v)$ with $i\in [k-1]$.

A \emph{vertex cover} of a $k$-uniform hypergraph $H=(V(H),E(H))$ is a subset $C$ of $V(H)$ with the property 
that $e\cap C\not= \emptyset$ holds for all $e\in E(H)$.
The Vertex Cover problem consists of finding a vertex cover of minimum size in a given 
$k$-uniform hypergraph. The Vertex Cover problem in $k$-partite $k$-uniform hypergraphs
is the restricted problem, where a $k$-partite $k$-uniform hypergraph and its
vertex partition is given as a part of the input.  
 
We define a $k$-partite $k$-uniform hypergraph $H=(V_1,..,V_k,E(H))$ as {\em $\eps$-dense} 
for an $\eps\in [0,1]$ if the following condition holds:
$$ |E(H)| \quad\geq\quad \epsilon \prod_{i\in [k]} |V_i|  $$

For $\ell\in [k-1]$, we introduce the notion of $\ell$-wise $\eps$-dense $k$-partite $k$-uniform hypergraphs. 
Given a $k$-partite $k$-uniform hypergraph $H$, if there exists an $I\in {[k] \choose \ell}$ and 
an $\eps\in [0,1]$ such that for all $S$ with the property $|V_i\cap S|=1$  for all $i\in I$ the condition
$$d_H(S)\quad\geq\quad \eps \prod_{i\in [k]\setminus I} |V_i|  $$
holds, we define $H$ to be $\ell$-wise $\eps$-dense. 

\section{Our Results}
In this paper, we give an improved approximation upper bound for the Vertex Cover problem restricted to $\epsilon$-dense
$k$-partite $k$-uniform hypergraphs. The approximation algorithm
in~\cite{CKSV10b} yields an approximation ratio of 
$$\frac{k}{k-(k-2)(1-\epsilon)^{\frac1{k-\ell}}}$$  for $\ell$-wise $\eps$-dense
balanced $k$-partite $k$-uniform hypergraphs. Here, we design an algorithm 
with an approximation factor of $$\frac{k}{2+(k-2) \epsilon}$$
for the $\eps$-dense case which also improves on the  
$\ell$-wise $\eps$-dense balanced case for all $\ell\in [k-2]$
and matches their bound when $\ell= k-1$. A further
advantage of this algorithm is that it applies to 
a larger class of hypergraphs since the considered hypergraph
is  not necessarily required to be balanced.\\
As a byproduct, we obtain a constructive proof that a
vertex cover of an  
 $\epsilon$-dense
$k$-partite $k$-uniform hypergraph $H=(V_1,..,V_k,E(H))$
is bounded from below by $\epsilon |V_k|$, which is shown to be 
sharp by constructing a family of tight examples.

On the other hand, we provide 
inapproximability results for the Vertex Cover problem 
restricted to $\ell$-wise $\epsilon$-dense balanced
$k$-partite $k$-uniform hypergraphs under the Unique Games Conjecture.  
We also prove 
that this reduction yields a matching lower bound if we use a
conjecture on the Unique Games hardness of the Vertex Cover
problem restricted to balanced $k$-partite $k$-uniform hypergraphs.
This means that further restrictions such as $\ell$-wise density
cannot lead to improved approximation ratios and our
proposed approximation algorithm is best possible assuming this conjecture.
In addition, we are able to prove an inapproximability factor
under $P\neq NP$.

\section{Approximation Algorithm }\label{approxopt}
In this section, we give a polynomial time approximation algorithm
with improved approximation factor for the 
Vertex Cover problem restricted to $\epsilon$-dense $k$-partite $k$-uniform hypergraphs.\\ 

We state now our main result.

\begin{theorem}\label{thm:main1}
There exists a polynomial time approximation
algorithm with  approximation ratio 
$$\frac{ k }{2+(k-2) \epsilon} $$ 
for the Vertex Cover problem in  $\eps$-dense
$k$-partite $k$-uniform hypergraphs.
\end{theorem} 

A crucial ingredient of the proof of Theorem~\ref{thm:main1} 
is  Lemma~\ref{lem:extract}, in which we show that we can extract
efficiently a large part of an optimal vertex cover of a given 
$\epsilon$-dense $k$-partite $k$-uniform hypergraph $H=(V_1,..,V_k,E(H))$.
More precisely, we obtain in this way a constructive proof
that the size of a vertex cover of $H$ is bounded from below
by $\epsilon |V_k|$. The procedure for the extraction of a part
of an optimal vertex cover is given in Figure~\ref{fig:extract}. 

\begin{figure}[ht]
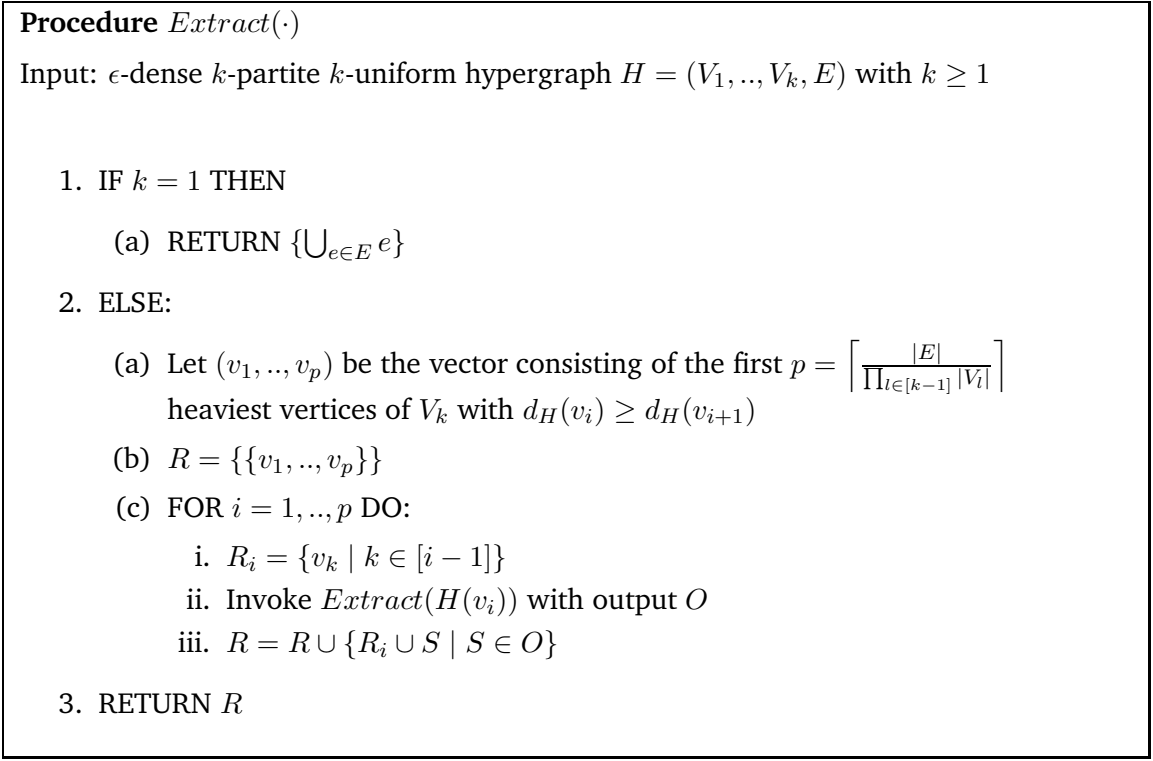

\noindent\fbox{
\begin{minipage}{\textwidth - 4mm}
\small
\noindent\textbf{Procedure $Extract(\cdot)$} \\[1ex]
Input: $\epsilon$-dense $k$-partite $k$-uniform hypergraph $H = (V_1,..,V_k,E)$
with $k\geq 1$\\[.8ex]
\begin{enumerate}
\item IF $k=1$ THEN 
\begin{enumerate}
\item RETURN $\{\bigcup_{e\in E}e\}$ 
\end{enumerate}
\item ELSE:
\begin{enumerate}
\item Let $(v_1,..,v_p)$ be the vector consisting  of the 
first $p=\left\lceil \frac{|E|}{\prod_{l\in [k-1]}|V_l|}\right\rceil$ \\ heaviest vertices of $V_k$ with $d_{H}(v_i)\geq d_{H}(v_{i+1})$ 
\item $R=\{\{v_1,..,v_p\} \}$
\item FOR $i=1,..,p$ DO: 
\begin{enumerate}
\item  $R_i=\{v_k \mid k\in [i-1] \}$
\item Invoke $Extract(H(v_i))$ with output $O$
\item $R=R\cup \{R_i\cup S \mid S\in O\}$
\end{enumerate}
\end{enumerate}
\item RETURN $R$\\
\end{enumerate}

\end{minipage}
}
\caption{Procedure $Extract$}
\label{fig:extract}
\end{figure}
We  now formulate Lemma~\ref{lem:extract}:

\begin{lemma}\label{lem:extract}
Let $H=(V_1,..,V_k,E(H))$ be an $\epsilon$-dense $k$-partite $k$-uniform hypergraph 
with $k\geq 1$. Then, the procedure $Extract(\cdot )$
computes in polynomial time a collection $R$ of subsets of $V(H)$ such that
the size of $R$ is polynomial in $|V(H)|$ and $R$ contains a set $S$, which is a subset
of an optimal vertex cover of $H$ and its cardinality is
at least $\epsilon |V_k|$. 
\end{lemma}
As a consequence, we obtain directly:
\begin{corollary}\label{corollary:extract}
Given  an $\epsilon$-dense $k$-partite $k$-uniform hypergraph $H=(V_1,..,V_k,E(H))$
with $k\geq 1$, the cardinality of an optimal vertex cover of $H$ is bounded from below by $\epsilon |V_k|$.  
\end{corollary}
Before we prove Lemma~\ref{lem:extract}, we describe the main idea of the proof.
Let $OPT$ denote an optimal vertex cover of $H$. The procedure $Extract(\cdot)$
tests for the set $R=\{v_1,..,v_p\}$ of the $p$ heaviest vertices of $V_k$,
if $\{v_1,..,v_{u-1}\}\subseteq OPT$ and $v_u\not \in OPT$ for every $u\in [p]$.
Clearly, either $R\subseteq OPT$ or there exists a $v_u$ such that $v_u\not \in OPT$.
If the procedure already possesses a part of $OPT$ denoted by $R_u$, then, $Extract(\cdot)$
tries to obtain a large part of an optimal vertex cover of the $v_u$-induced hypergraph
$H(v_u)$. Hence, we have to show that $H(v_u)$ must still be dense enough.\\ 
We now give the proof of Lemma~\ref{lem:extract}.
\begin{proof}
The proof of Lemma~\ref{lem:extract} will be split in several parts.
In particular, we show that given an $\eps$-dense $k$-partite $k$-uniform 
hypergraph $H=(V_1,..,V_k,E(H))$, the procedure
$Extract(\cdot )$ and its output $R$ possess the following properties:
\begin{enumerate}
\item $Extract(\cdot )$ constructs $R$ in polynomial time and the cardinality of $R$ is $O(n^k)$.
\item There is a $S\in R$ such that $S$ is a subset of an optimal vertex cover of $H$.
\item For every $S\in R$, the cardinality of $S$ is at least $|S|\geq \epsilon |V_k|$.
\end{enumerate}
(1.) Clearly, $R$ is upper bounded by $|V_1|^k=O(n^k)$ and 
therefore, the running time of $Extract(\cdot)$ is 
$O(n^k)$.\\
(2.) and (3.) We prove the remaining properties by induction.
If we have $k=1$, the set $\bigcup_{e\in E(H)}e$ is by definition
an optimal vertex cover of $H=(V_1,E(H))$. Since $H$ is
$\epsilon$-dense,  the cardinality of $|E(H)|$ is lower bounded by
$\epsilon|V_1|$.\\
We assume that $k>1$.     
Let $H=(V_1,..,V_k,E(H))$ be an $\epsilon$-dense 
$k$-partite $k$-uniform hypergraph and $OPT\subseteq V(H)$ an
optimal vertex  cover of $H$. 
Let $(v_1,..,v_p)$ be the vector consisting  of the 
first $p=\left\lceil \frac{|E(H)|}{\prod_{l\in [k-1]}|V_l|}\right\rceil$ 
heaviest vertices of $V_k$ with $d_{H}(v_i)\geq d_{H}(v_{i+1})$.
If $\{v_1,..,v_p\}$ is contained
in $OPT$, we have constructed a subset of an optimal vertex cover
with cardinality 
$$p=\left\lceil \frac{|E(H)|}{\prod\limits_{l\in [k-1]}|V_l|}\right\rceil\quad\geq\quad 
\frac{\epsilon \prod\limits_{l\in [k]}|V_l|}{\prod\limits_{l\in [k-1]}|V_l|} 
\quad\geq\quad \epsilon|V_k|.$$
Otherwise, there is an $u\in [p]$ such that $R_u\subseteq OPT$ 
and $v_u\not \in OPT$. But this means that an optimal vertex cover
of $H$ contains an optimal vertex cover of the $v_u$-induced
$(k-1)$-partite $(k-1)$-uniform hypergraph $H(v_u)$
in order to cover the edges $e\in \{e\in E \mid v_u\in e\}$.
The situation is depicted in Figure~\ref{fig:inducedex}.\\

\begin{figure}[h]
\begin{center}
\setlength{\unitlength}{0.00083333in}
\begingroup\makeatletter\ifx\SetFigFont\undefined%
\gdef\SetFigFont#1#2#3#4#5{%
  \reset@font\fontsize{#1}{#2pt}%
  \fontfamily{#3}\fontseries{#4}\fontshape{#5}%
  \selectfont}%
\fi\endgroup%
{\renewcommand{\dashlinestretch}{30}
\begin{picture}(6399,3034)(0,-10)
\put(612,1948){\blacken\ellipse{90}{90}}
\put(612,1948){\ellipse{90}{90}}
\put(612,2248){\blacken\ellipse{90}{90}}
\put(612,2248){\ellipse{90}{90}}
\put(612,1348){\blacken\ellipse{90}{90}}
\put(612,1348){\ellipse{90}{90}}
\put(2262,1048){\blacken\ellipse{90}{90}}
\put(2262,1048){\ellipse{90}{90}}
\put(2262,1348){\blacken\ellipse{90}{90}}
\put(2262,1348){\ellipse{90}{90}}
\put(4437,1048){\blacken\ellipse{90}{90}}
\put(4437,1048){\ellipse{90}{90}}
\put(4437,1348){\blacken\ellipse{90}{90}}
\put(4437,1348){\ellipse{90}{90}}
\put(5787,1048){\blacken\ellipse{90}{90}}
\put(5787,1048){\ellipse{90}{90}}
\put(5787,1348){\blacken\ellipse{90}{90}}
\put(5787,1348){\ellipse{90}{90}}
\put(5787,2248){\blacken\ellipse{90}{90}}
\put(5787,2248){\ellipse{90}{90}}
\put(4437,2248){\blacken\ellipse{90}{90}}
\put(4437,2248){\ellipse{90}{90}}
\put(2262,2323){\blacken\ellipse{90}{90}}
\put(2262,2323){\ellipse{90}{90}}
\put(2262,1768){\blacken\ellipse{90}{90}}
\put(2262,1768){\ellipse{90}{90}}
\put(612,1123){\blacken\ellipse{90}{90}}
\put(612,1123){\ellipse{90}{90}}
\put(4437,1798){\blacken\ellipse{90}{90}}
\put(4437,1798){\ellipse{90}{90}}
\put(5787,1873){\blacken\ellipse{90}{90}}
\put(5787,1873){\ellipse{90}{90}}
\dottedline{45}(462,1723)(462,1573)
\put(117,478){\arc{210}{1.5708}{3.1416}}
\put(117,2368){\arc{210}{3.1416}{4.7124}}
\put(807,2368){\arc{210}{4.7124}{6.2832}}
\put(807,478){\arc{210}{0}{1.5708}}
\path(12,478)(12,2368)
\path(117,2473)(807,2473)
\path(912,2368)(912,478)
\path(807,373)(117,373)
\thicklines
\dottedline{68}(3162,1423)(3612,1423)
\thinlines
\put(1917,478){\arc{210}{1.5708}{3.1416}}
\put(1917,2368){\arc{210}{3.1416}{4.7124}}
\put(2607,2368){\arc{210}{4.7124}{6.2832}}
\put(2607,478){\arc{210}{0}{1.5708}}
\path(1812,478)(1812,2368)
\path(1917,2473)(2607,2473)
\path(2712,2368)(2712,478)
\path(2607,373)(1917,373)
\put(4092,478){\arc{210}{1.5708}{3.1416}}
\put(4092,2368){\arc{210}{3.1416}{4.7124}}
\put(4782,2368){\arc{210}{4.7124}{6.2832}}
\put(4782,478){\arc{210}{0}{1.5708}}
\path(3987,478)(3987,2368)
\path(4092,2473)(4782,2473)
\path(4887,2368)(4887,478)
\path(4782,373)(4092,373)
\put(5442,478){\arc{210}{1.5708}{3.1416}}
\put(5442,2368){\arc{210}{3.1416}{4.7124}}
\put(6132,2368){\arc{210}{4.7124}{6.2832}}
\put(6132,478){\arc{210}{0}{1.5708}}
\path(5337,478)(5337,2368)
\path(5442,2473)(6132,2473)
\path(6237,2368)(6237,478)
\path(6132,373)(5442,373)
\dashline{60.000}(612,1123)(1812,2623)(6387,2623)
	(6387,898)(1737,898)(612,1123)
\path(2362,1623)(2401,1629)(2441,1635)
	(2481,1640)(2521,1644)(2562,1649)
	(2604,1653)(2646,1656)(2689,1659)
	(2733,1662)(2777,1664)(2821,1666)
	(2867,1668)(2913,1669)(2959,1670)
	(3006,1671)(3054,1671)(3102,1672)
	(3150,1671)(3198,1671)(3247,1670)
	(3296,1669)(3344,1668)(3393,1667)
	(3441,1665)(3490,1664)(3538,1662)
	(3585,1660)(3633,1658)(3679,1655)
	(3725,1653)(3771,1651)(3816,1648)
	(3860,1646)(3903,1644)(3946,1641)
	(3988,1639)(4029,1637)(4070,1634)
	(4110,1632)(4149,1630)(4187,1628)
	(4225,1626)(4263,1625)(4300,1623)
	(4344,1621)(4388,1620)(4432,1618)
	(4476,1617)(4519,1616)(4563,1615)
	(4606,1614)(4649,1614)(4693,1614)
	(4736,1614)(4778,1614)(4821,1614)
	(4863,1614)(4906,1615)(4947,1616)
	(4988,1617)(5029,1618)(5069,1619)
	(5108,1621)(5146,1622)(5183,1624)
	(5220,1626)(5255,1628)(5289,1630)
	(5322,1632)(5354,1634)(5385,1637)
	(5414,1639)(5443,1642)(5470,1644)
	(5496,1647)(5521,1649)(5545,1652)
	(5568,1655)(5591,1658)(5612,1660)
	(5649,1666)(5683,1671)(5717,1677)
	(5748,1683)(5778,1689)(5806,1696)
	(5833,1703)(5858,1711)(5881,1719)
	(5902,1726)(5922,1735)(5939,1743)
	(5954,1751)(5968,1760)(5979,1768)
	(5989,1777)(5996,1785)(6003,1794)
	(6008,1802)(6012,1810)(6016,1823)
	(6018,1836)(6018,1849)(6016,1864)
	(6012,1878)(6007,1893)(6000,1907)
	(5991,1921)(5982,1935)(5971,1947)
	(5960,1958)(5948,1969)(5937,1978)
	(5925,1985)(5915,1991)(5905,1996)
	(5894,2000)(5881,2004)(5868,2007)
	(5853,2010)(5837,2013)(5819,2015)
	(5799,2016)(5777,2017)(5753,2017)
	(5728,2017)(5700,2017)(5671,2016)
	(5641,2015)(5608,2014)(5574,2012)
	(5537,2010)(5516,2009)(5495,2008)
	(5472,2007)(5449,2006)(5424,2005)
	(5398,2003)(5371,2002)(5342,2001)
	(5312,1999)(5281,1998)(5249,1997)
	(5215,1995)(5180,1994)(5144,1993)
	(5107,1991)(5069,1990)(5029,1989)
	(4989,1988)(4948,1987)(4907,1986)
	(4864,1985)(4822,1984)(4778,1984)
	(4735,1983)(4691,1983)(4647,1982)
	(4603,1982)(4558,1982)(4514,1982)
	(4469,1983)(4424,1983)(4379,1984)
	(4333,1985)(4287,1985)(4251,1986)
	(4214,1987)(4177,1988)(4139,1989)
	(4101,1990)(4061,1992)(4021,1993)
	(3981,1994)(3939,1995)(3897,1996)
	(3854,1998)(3810,1999)(3765,2000)
	(3720,2001)(3674,2002)(3628,2003)
	(3581,2003)(3533,2004)(3485,2004)
	(3437,2004)(3388,2004)(3339,2004)
	(3290,2003)(3241,2003)(3192,2002)
	(3143,2001)(3095,1999)(3046,1997)
	(2998,1995)(2950,1993)(2903,1990)
	(2856,1987)(2810,1984)(2764,1980)
	(2719,1976)(2675,1972)(2631,1967)
	(2587,1962)(2544,1957)(2502,1951)
	(2460,1944)(2419,1938)(2378,1931)
	(2337,1923)(2292,1914)(2248,1905)
	(2204,1895)(2159,1884)(2115,1873)
	(2070,1861)(2025,1848)(1981,1835)
	(1936,1822)(1891,1807)(1846,1793)
	(1800,1778)(1755,1762)(1710,1746)
	(1665,1729)(1621,1712)(1576,1695)
	(1533,1678)(1489,1660)(1446,1643)
	(1404,1625)(1363,1607)(1322,1590)
	(1282,1572)(1244,1554)(1206,1537)
	(1170,1520)(1134,1503)(1100,1487)
	(1067,1471)(1035,1455)(1004,1440)
	(975,1425)(946,1411)(919,1397)
	(893,1384)(868,1371)(844,1359)
	(821,1347)(800,1335)(758,1314)
	(720,1294)(684,1276)(651,1258)
	(621,1242)(592,1227)(567,1213)
	(543,1200)(522,1187)(504,1176)
	(487,1166)(473,1156)(462,1147)
	(452,1139)(444,1132)(438,1125)
	(433,1118)(429,1111)(427,1105)
	(425,1098)(423,1088)(422,1077)
	(422,1065)(423,1053)(425,1040)
	(429,1027)(433,1014)(439,1001)
	(445,990)(452,979)(460,969)
	(469,961)(477,954)(487,948)
	(499,943)(512,940)(526,938)
	(541,938)(557,939)(573,941)
	(588,943)(603,947)(617,950)
	(629,954)(640,957)(650,960)
	(655,962)(660,964)(666,966)
	(671,968)(678,971)(685,975)
	(695,979)(705,984)(718,990)
	(733,998)(751,1007)(770,1017)
	(793,1029)(818,1042)(846,1056)
	(877,1073)(911,1091)(950,1110)
	(969,1121)(990,1131)(1011,1143)
	(1034,1154)(1058,1167)(1083,1180)
	(1110,1193)(1137,1207)(1166,1221)
	(1197,1236)(1229,1251)(1261,1266)
	(1296,1282)(1331,1298)(1368,1314)
	(1406,1331)(1444,1347)(1484,1364)
	(1525,1380)(1566,1397)(1608,1413)
	(1651,1429)(1694,1445)(1738,1461)
	(1782,1476)(1826,1491)(1871,1505)
	(1915,1519)(1960,1532)(2005,1544)
	(2049,1556)(2094,1568)(2138,1579)
	(2183,1589)(2227,1598)(2272,1607)
	(2317,1615)(2362,1623)
\put(162,1048){\makebox(0,0)[lb]{\smash{{\SetFigFont{12}{14.4}{\rmdefault}{\mddefault}{\updefault}$v_u$}}}}
\put(162,1348){\makebox(0,0)[lb]{\smash{{\SetFigFont{12}{14.4}{\rmdefault}{\mddefault}{\updefault}$v_{u-1}$}}}}
\put(162,1873){\makebox(0,0)[lb]{\smash{{\SetFigFont{12}{14.4}{\rmdefault}{\mddefault}{\updefault}$v_2$}}}}
\put(162,2173){\makebox(0,0)[lb]{\smash{{\SetFigFont{12}{14.4}{\rmdefault}{\mddefault}{\updefault}$v_1$}}}}
\put(312,73){\makebox(0,0)[lb]{\smash{{\SetFigFont{12}{14.4}{\rmdefault}{\mddefault}{\updefault}$V_k$}}}}
\put(2112,73){\makebox(0,0)[lb]{\smash{{\SetFigFont{12}{14.4}{\rmdefault}{\mddefault}{\updefault}$V_{k-1}$}}}}
\put(4287,73){\makebox(0,0)[lb]{\smash{{\SetFigFont{12}{14.4}{\rmdefault}{\mddefault}{\updefault}$V_2$}}}}
\put(5712,73){\makebox(0,0)[lb]{\smash{{\SetFigFont{12}{14.4}{\rmdefault}{\mddefault}{\updefault}$V_1$}}}}
\put(2562,2848){\makebox(0,0)[lb]{\smash{{\SetFigFont{12}{14.4}{\familydefault}{\mddefault}{\updefault}\textbf{$v_u$-induced Hypergraph $H(v_u)$}}}}}
\end{picture}
}
\end{center}
\caption{The $v_u$-induced  $(k-1)$-partite $(k-1)$-uniform hypergraph $H(v_u)$}
\label{fig:inducedex}
\end{figure}
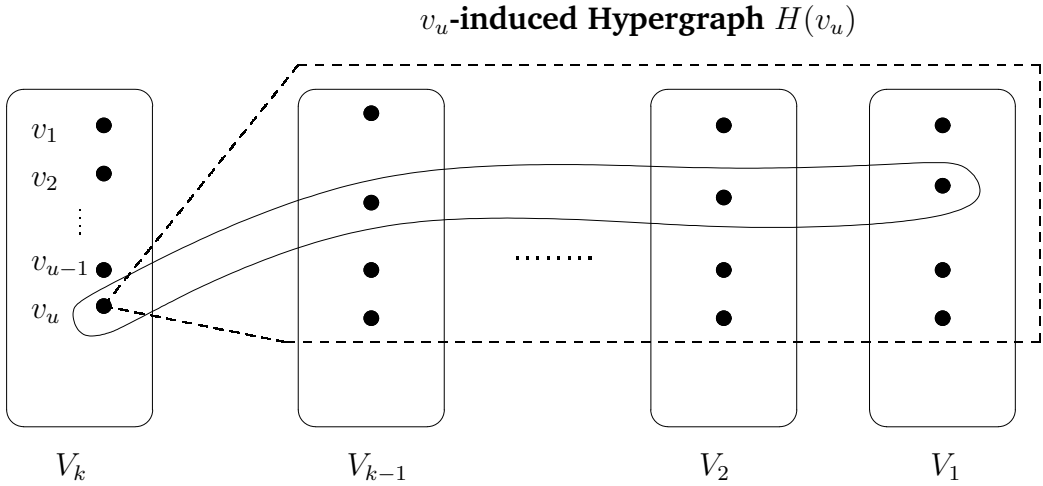
By our induction hypothesis, $Extract(H(v_u))$ contains
a set $S_u$ which is a subset of a minimum vertex cover
of $H(v_u)$ and of $OPT$. The only claim, which remains to
be proven, is that the cardinality of $S_u$ is large enough.
More precisely, we show that $|S_u|$ can be lower bounded
by $\epsilon|V_k|-|R_u|$. Therefore, we need to analyze
the density of the $v_u$-induced hypergraph $H(v_u)$. The edge 
set of $H(v_u)$ is given by $\{e\setminus\{v_u\}\mid v_u\in e\in E\}$.
Thus, we have to obtain a lower bound on the degree of $v_u$.
Since $|\{e\in E \mid e\cap R_u \neq \emptyset \}|$ is upper bounded
by $|R_u|\prod_{l\in [k-1]}|V_l|$, the vertices in $V_k\setminus R_u$ possess
the average degree of at least 
\begin{eqnarray}
\frac{\sum\limits_{v\in V_k\setminus R_u} deg_{H}(v)  }{|V_k\setminus R_u|}
& \geq &
\frac{\epsilon \prod\limits_{l\in [k]}|V_l|-|\{e\in E \mid e\cap R_u \neq \emptyset \}|}{|V_k\setminus R_u|} \\[1ex]
& \geq &
\frac{\epsilon \prod\limits_{l\in [k]}|V_l|-|R_u|\prod\limits_{l\in [k-1]}|V_l|}{|V_k\setminus R_u|} \\[1ex]
& \geq &
\frac{(\epsilon|V_k|-|R_u|) \prod\limits_{l\in [k-1]}|V_l|}{|V_k\setminus R_u|} 
\end{eqnarray}
Since the heaviest vertex in $V_k\setminus R_u $ must have a degree
of at least $\frac{(\epsilon|V_k|-|R_u|) \prod_{l\in [k-1]}|V_l|}{|V_k\setminus R_u|}$, 
we deduce that the edge set of $H(v_u)$ denoted by $E_u$ can be lower bounded  by
$$|E_u| \quad\geq\quad \frac{(\epsilon|V_k|-|R_u|) \prod\limits_{l\in [k-1]}|V_l|}{|V_k\setminus R_u|} $$
Let $H(v_u)$ be defined by $(V^u_1,..,V^u_{k-1},E_u)$ with $|V^u_i| \leq |V_i|$ 
for all $i\in [k-1]$. By our induction hypothesis, the size of every set contained in $Extract(\cdot)$
is at least
 \begin{eqnarray}
 \frac{ |E_u| }{\prod\limits_{l\in [k-1]}|V^u_l| }|V_{k-1}|
& \geq &
\frac{(\epsilon|V_k|-|R_u|) \prod\limits_{l\in [k-1]}|V_l|}{|V_k\setminus R_u|\prod\limits_{l\in [k-1]}|V^u_l|}|V_{k-1}| \\[1ex]
& \geq &
\frac{(\epsilon|V_k|-|R_u|) \prod\limits_{l\in [k-1]}|V_l|}{|V_k\setminus R_u|\prod\limits_{l\in [k-1]}|V_l|}|V_{k-1}| \\[1ex]
& \geq &
\frac{(\epsilon|V_k|-|R_u|) }{|V_k\setminus R_u|}|V_{k}| \\[1ex]
& \geq &
\frac{(\epsilon|V_k|-|R_u|) }{|V_k|}|V_{k}| \,=\, \epsilon|V_k|-|R_u|
\end{eqnarray}

In $(4)$, we used the fact that $|V^u_i| \leq |V_i|$ 
for all $i\in [k-1]$. Whereas in $(5)$, we used our 
assumption $|V_k| \leq |V_{k-1}|$. 
All in all, we obtain
\begin{eqnarray}
|R_u \cup S_u|
& \geq &
|R_u|+(|\epsilon|V_k|-|R_u|)\;=\;\epsilon|V_k|. 
\end{eqnarray}
Clearly, this argumentation on the size of $R_u\cup S_u$ holds for every $u\in [p]$ and 
the proof of Lemma~\ref{lem:extract} follows.  
\end{proof}
Before we state our approximation algorithm and prove Theorem~\ref{thm:main1}, we show 
that the bound in Lemma~\ref{lem:extract} is tight.
In particular, we define a family of $\epsilon$-dense $k$-partite $k$-uniform hypergraphs $H(k,l,\epsilon)=(V_1,..,V_k,E(H_l))$ with 
$|V_i|=\frac{|V|}{k}$ for all $i\in [k]$, $k\geq 1$, $\epsilon \in \{\frac ul\mid u\in [l]\}$ and $l\geq 1$ such that $Extract(\cdot)$ returns a subset 
of an optimal vertex cover with cardinality of exactly $\epsilon |V_k|$.   

\begin{lemma}\label{lem:tight}
The bound of Lemma~\ref{lem:extract} is tight.
\end{lemma}
\begin{proof}
Let us define $H(k,p,\epsilon)=(V_1,..,V_k,E)$. For a fixed $p\geq 1$ and $k\geq 1$, 
every partition $V_i$ with $i\in [k]$ consists of a set of $l$ vertices. Let us fix a $\epsilon=\frac{u}{l} $ with 
$u\in  [l]$. Then, $H(k,l,\epsilon)$ contains the set $V^u_k\subseteq V_k$ of $u$ vertices such that
$E=\{\{v_1,v_2,..,v_k  \} \mid v_1\in V^u_k, v_2\in V_2,..,v_k\in V_k       \}$.
An example of such a hypergraph is depicted in Figure~\ref{fig:example}.\\

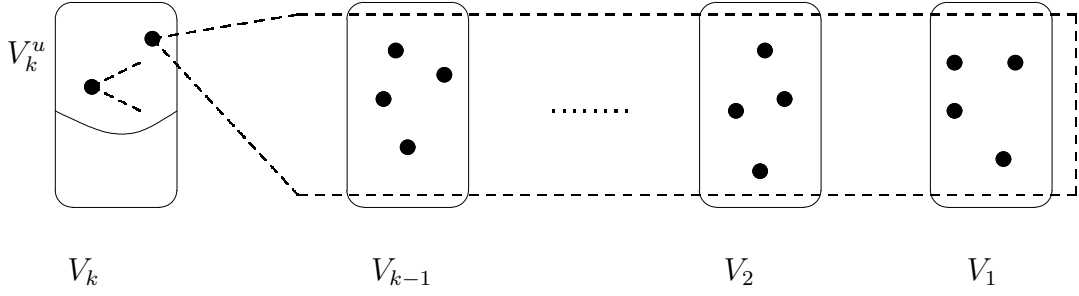
\begin{figure}[h]
\begin{center}
\setlength{\unitlength}{0.00083333in}
\begingroup\makeatletter\ifx\SetFigFont\undefined%
\gdef\SetFigFont#1#2#3#4#5{%
  \reset@font\fontsize{#1}{#2pt}%
  \fontfamily{#3}\fontseries{#4}\fontshape{#5}%
  \selectfont}%
\fi\endgroup%
{\renewcommand{\dashlinestretch}{30}
\begin{picture}(6627,1825)(0,-10)
\put(4515,1123){\blacken\ellipse{90}{90}}
\put(4515,1123){\ellipse{90}{90}}
\put(5865,1123){\blacken\ellipse{90}{90}}
\put(5865,1123){\ellipse{90}{90}}
\put(5865,1423){\blacken\ellipse{90}{90}}
\put(5865,1423){\ellipse{90}{90}}
\put(2340,1198){\blacken\ellipse{90}{90}}
\put(2340,1198){\ellipse{90}{90}}
\put(4695,1498){\blacken\ellipse{90}{90}}
\put(4695,1498){\ellipse{90}{90}}
\put(6240,1423){\blacken\ellipse{90}{90}}
\put(6240,1423){\ellipse{90}{90}}
\put(915,1573){\blacken\ellipse{90}{90}}
\put(915,1573){\ellipse{90}{90}}
\put(4815,1198){\blacken\ellipse{90}{90}}
\put(4815,1198){\ellipse{90}{90}}
\put(540,1273){\blacken\ellipse{90}{90}}
\put(540,1273){\ellipse{90}{90}}
\put(2415,1498){\blacken\ellipse{90}{90}}
\put(2415,1498){\ellipse{90}{90}}
\put(2490,898){\blacken\ellipse{90}{90}}
\put(2490,898){\ellipse{90}{90}}
\put(2715,1348){\blacken\ellipse{90}{90}}
\put(2715,1348){\ellipse{90}{90}}
\put(4665,748){\blacken\ellipse{90}{90}}
\put(4665,748){\ellipse{90}{90}}
\put(6165,823){\blacken\ellipse{90}{90}}
\put(6165,823){\ellipse{90}{90}}
\put(420,628){\arc{210}{1.5708}{3.1416}}
\put(420,1693){\arc{210}{3.1416}{4.7124}}
\put(960,1693){\arc{210}{4.7124}{6.2832}}
\put(960,628){\arc{210}{0}{1.5708}}
\path(315,628)(315,1693)
\path(420,1798)(960,1798)
\path(1065,1693)(1065,628)
\path(960,523)(420,523)
\put(2220,628){\arc{210}{1.5708}{3.1416}}
\put(2220,1693){\arc{210}{3.1416}{4.7124}}
\put(2760,1693){\arc{210}{4.7124}{6.2832}}
\put(2760,628){\arc{210}{0}{1.5708}}
\path(2115,628)(2115,1693)
\path(2220,1798)(2760,1798)
\path(2865,1693)(2865,628)
\path(2760,523)(2220,523)
\put(4395,628){\arc{210}{1.5708}{3.1416}}
\put(4395,1693){\arc{210}{3.1416}{4.7124}}
\put(4935,1693){\arc{210}{4.7124}{6.2832}}
\put(4935,628){\arc{210}{0}{1.5708}}
\path(4290,628)(4290,1693)
\path(4395,1798)(4935,1798)
\path(5040,1693)(5040,628)
\path(4935,523)(4395,523)
\put(5820,628){\arc{210}{1.5708}{3.1416}}
\put(5820,1693){\arc{210}{3.1416}{4.7124}}
\put(6360,1693){\arc{210}{4.7124}{6.2832}}
\put(6360,628){\arc{210}{0}{1.5708}}
\path(5715,628)(5715,1693)
\path(5820,1798)(6360,1798)
\path(6465,1693)(6465,628)
\path(6360,523)(5820,523)
\dashline{60.000}(915,1573)(1815,1723)(6615,1723)
	(6615,598)(1815,598)(915,1573)
\thicklines
\dottedline{68}(3390,1123)(3840,1123)
\thinlines
\dashline{60.000}(540,1273)(840,1123)
\dashline{60.000}(540,1273)(840,1423)
\path(315,1123)(316,1123)(319,1121)
	(328,1117)(343,1109)(364,1099)
	(389,1088)(416,1075)(444,1062)
	(470,1050)(494,1039)(516,1029)
	(537,1021)(555,1014)(571,1008)
	(587,1003)(603,998)(620,993)
	(637,989)(654,986)(672,983)
	(690,982)(708,980)(726,980)
	(744,980)(761,982)(778,983)
	(794,986)(810,989)(825,993)
	(840,998)(853,1003)(866,1008)
	(880,1014)(894,1021)(910,1029)
	(928,1039)(947,1050)(967,1062)
	(988,1075)(1009,1088)(1028,1099)
	(1044,1109)(1055,1117)(1062,1121)(1065,1123)
\put(390,73){\makebox(0,0)[lb]{\smash{{\SetFigFont{12}{14.4}{\rmdefault}{\mddefault}{\updefault}$V_k$}}}}
\put(2265,73){\makebox(0,0)[lb]{\smash{{\SetFigFont{12}{14.4}{\rmdefault}{\mddefault}{\updefault}$V_{k-1}$}}}}
\put(4440,73){\makebox(0,0)[lb]{\smash{{\SetFigFont{12}{14.4}{\rmdefault}{\mddefault}{\updefault}$V_2$}}}}
\put(5940,73){\makebox(0,0)[lb]{\smash{{\SetFigFont{12}{14.4}{\rmdefault}{\mddefault}{\updefault}$V_1$}}}}
\put(15,1423){\makebox(0,0)[lb]{\smash{{\SetFigFont{12}{14.4}{\rmdefault}{\mddefault}{\updefault}$V^u_k$}}}}
\end{picture}
}
\end{center}
\caption{An example of a hypergraph $H(k,l,\epsilon)$}
\label{fig:example}
\end{figure} 
\noindent Notice that $H(k,l,\epsilon)=(V_1,..,V_k,E)$
is $\epsilon$-dense, since 
$$\frac{|E|}{\prod\limits_{j\in [k]}|V_j|}= \frac{|V^u_k|}{|V_k|}=\frac{u}{l}=\epsilon .$$ 
The procedure $Extract(\cdot)$ returns a set $R$, in which $V^u_k$ is contained, since $V^u_k$
is the set of the $p$ heaviest vertices of $V_k$. Hence, we obtain $|V^u_k|=\frac{|V^u_k|}{|V_k|}|V_k|=\epsilon|V_k|$.
On the other hand, the remaining hypergraph $H'=(V_1,..,V_k\setminus V^u_k,E(H') )$ with 
edge set $E(H')=\{e\in E \mid e\cap V^u_k= \emptyset\}$ is already covered, since $E(H')$ is
by definition of $H(k,p,\epsilon)$ the empty set. Therefore, $V^u_k$ is a vertex cover of 
$H(k,p,\epsilon)$ and  since, according to Corollary~\ref{corollary:extract}, every 
vertex cover is bounded from below
by $\epsilon |V_k|$, $V^u_k$ must be an optimal vertex cover.     
\end{proof}

Next, we state our approximation algorithm for the Vertex Cover problem 
in $\epsilon$-dense $k$-partite $k$-uniform hypergraphs defined in
Figure~\ref{fig:approx}. The approximation
algorithm combines the procedure $Extract(\cdot )$ to generate a large enough
subset of an optimal vertex cover together with the $\frac k2$-approximation algorithm
due to Lov\'asz~\cite{L75} applied to the remaining instance.\\
\bigskip

\begin{figure}[ht]
\noindent\fbox{
\begin{minipage}{\textwidth - 4mm}
\small
\noindent\textbf{Algorithm $Approx(\cdot)$} \\[1ex]
Input: $\epsilon$-dense $k$-partite $k$-uniform hypergraph $H = (V_1,..,V_k,E)$
with $k\geq 3$\\[.5ex]
\begin{enumerate}
\item  $T=\{V_k\}$
\item invoke procedure $Extract(H)$ with output $R$ 
\item for all $S\in R$ do : 
\begin{enumerate}
\item  $H_S=(V(H)\setminus S, \{e\in E(H)\mid e\cap S=\emptyset\})$
\item obtain a $(\frac k2)$-approximate solution $S_k$ for $H_S$
\item $T=T\cup \{S_k\cup S \}$
\end{enumerate}
\item Return the smallest set in $T$\\
\end{enumerate}
\end{minipage}
}
\caption{Algorithm $Approx(\cdot)$}
\label{fig:approx}
\end{figure}

\newpage

We now prove Theorem~\ref{thm:main1}.
\begin{proof}
Let $H=(V_1,..,V_k,E)$ be an
$\eps$-dense $k$-partite $k$-uniform hypergraph.
From Lemma~\ref{lem:extract}, we know that the 
procedure $Extract(\cdot)$ returns in polynomial time a collection $C$ of subsets of $V(H)$ such that there
is a set $S$ in $C$, which is contained in an optimal vertex cover of $H$.
Moreover, we know that the size of $S$ is lower bounded by $\epsilon |V_k|$.\\ 
Next, we analyze the approximation ratio
of our approximation algorithm $Approx(\cdot)$. Clearly, the size of an optimal vertex cover
of $H$ is upper bounded by $|V_k|$. Let us denote by $OPT'$ the size of an optimal vertex
cover of the remaining hypergraph $H'$ defined by removing all edges $e$ of $H$ with 
$e\cap S\neq \emptyset$. Furthermore, let $S'$ be the solution of the $\frac k2$-approximation
algorithm applied to $H'$.  The approximation ratio of $Approx(\cdot)$ is bounded by
\begin{eqnarray}
\frac{|S|+|S'|  }   {|S|+|OPT'|}\,\,\leq\,\, \frac{|S|+\frac k2|OPT'|}{|S|+|OPT'|}
&\leq & \frac{k}{ \frac{k|S|+k|OPT'|}{|S|+\frac k2|OPT'|}    } \\
& \leq  &   \frac{k}{ \frac{2|S|+(k-2)|S|+k|OPT'|}{|S|+\frac k2|OPT'|}    }  \\
 & \leq  & \frac{k}{2+(k-2) \frac{|S|}{|S|+\frac k2|OPT'|}    }  \\
 &  \leq  & \frac{k}{2+(k-2) \frac{|S|}{|V_k|}    }  \\
 & \leq  &   \frac{k}{2+(k-2) \frac{\epsilon|V_k|}{|V_k|} }\\
   & \leq  &   \frac{k}{2+(k-2) \epsilon} 
\end{eqnarray}  
In $(11)$, we used the fact that the size of the output of $Approx(\cdot)$
is upper bounded by $|V_k|$. Therefore, we have $|S|+\frac k2|OPT'|\leq |V_k|$. 
In $(12)$, we know from Lemma~\ref{lem:extract} that $|S|\geq \epsilon|V_k|$.
\end{proof}

\section{Inapproximability Results}
\label{sec:lb}

In this section, we prove hardness results for the Vertex Cover problem 
restricted to $\ell$-wise $\epsilon$-dense balanced $k$-uniform $k$-partite
hypergraphs under the Unique Games Conjecture~\cite{K02} as well as under the assumption $P\neq NP$.

\subsection{UGC-Hardness}
The Unique Games-hardness result of~\cite{GS10b} was obtained by applying the result
of Kumar et al.~\cite{KMT11}, with a modification to the LP integrality gap
due to Ahorani et al.~\cite{AHK96}.
More precisely, they proved the following inapproximability result:
  \begin{theorem}\cite{GS10b}\label{ugcmain2}
  For every $\delta>0$ and $k\geq 3$, there exist a $n_{\delta}$ such that given  
  $H=(V_1,..,V_k,E(H))$ as an instance of the 
Vertex Cover problem in balanced $k$-partite $k$-uniform hypergraphs with $|V(H)|\geq n_{\delta}$, 
the  following is UGC-hard to decide:
\begin{itemize}
\item The size of a vertex cover of $H$ is at least $|V|\left(\frac{1}{2(k-1)}-\delta\right)$.
\item The size of an optimal vertex cover of $H$ is at 
most $|V|\left(\frac{1}{k(k-1)}+\delta\right)$.
\end{itemize}
\end{theorem}
As the starting point of our reduction, we use Theorem~\ref{ugcmain2} and prove the following: 
\begin{theorem}\label{ugcmain3}
For every $\delta>0$, $\epsilon\in (0,1)$, $\ell\in [k-1]$, and $k\geq 3$, there exists no polynomial time approximation
algorithm with an approximation ratio 
$$\frac{k }{2+\frac{2(k-1)(k-2)\epsilon}{k +  (k- 2)\epsilon}}-\delta$$ 
for the Vertex Cover problem in  $\ell$-wise $\eps$-dense
$k$-partite $k$-uniform hypergraphs assuming the Unique Games Conjecture.
\end{theorem} 

\begin{proof}
First, we concentrate on the $\epsilon$-dense case and afterwards, 
we extend the range of $\ell$.
As a starting point of the reduction, we use the $k$-partite $k$-uniform hypergraph 
$H=(V_1,..,V_k,E(H))$ from Theorem~\ref{ugcmain2}
and construct an $\eps$-dense $k$-partite 
$k$-uniform hypergraph $H'=(V'_1,..,V'_k,E')$.\\
Let us start with the description of $H'$. First, we join the set $C_i$ of $\frac{\epsilon}{1-\epsilon}\frac{n}{k}$ vertices to $V_i$ 
for every $i\in [k]$ and add 
all possible edges $e$ of $H'$ to $E'$ with the restriction $C_1\cap e\neq \emptyset $.
Thus, we obtain $|V'_i|=\frac n k+\frac{\epsilon}{1-\epsilon}\frac{n}{k}$  for all 
$i\in [k]$.\\ 
Now, let us analyze how the size of the optimal solution of $H'$ transforms.
We denote by $OPT'$
an optimal vertex cover of $H'$. The UGC-hard decision question from Theorem~\ref{ugcmain2}
transforms into the following:
$$
n\left(\frac{1}{2(k-1)}-\delta\right )+\frac{\epsilon}{1-\epsilon}\frac{n}{k}\leq |OPT'|\,\,\textrm{  or  }\,\,
|OPT'|\leq n\left(\frac{1}{k(k-1)}+\delta\right)+\frac{\epsilon}{1-\epsilon}\frac{n}{k} \\
$$

Assuming the UGC, this implies the hardness of approximating the Vertex Cover problem
in  $\eps$-dense hypergraphs for every $\delta'>0$ to within:
\begin{eqnarray}
\frac{n\left(\frac{1}{2(k-1)}-\delta\right)+\frac{\epsilon}{1-\epsilon}\frac{n}{k}  }
{  n\left(\frac{1}{k(k-1)}+\delta\right)+\frac{\epsilon}{1-\epsilon}\frac{n}{k}}
& = & \frac{\frac{1-\epsilon}{2(k-1)}-\delta(1-\epsilon)+\frac{\epsilon}{k}  }
{  \frac{1-\epsilon}{k(k-1)}+\delta(1-\epsilon)+\frac{\epsilon}{k}} \\[1ex]
 & = & \frac{\frac{(1-\epsilon)k}{2(k-1)k}+
 \frac{2\epsilon(k-1)}{2k(k-1)}}{\frac{1-\epsilon}{(k-1)k}+\frac{\epsilon (k-1)}{k(k-1)} } -\delta'
\end{eqnarray}

\begin{eqnarray}
\frac{\frac{(1-\epsilon)k}{2(k-1)k}+
 \frac{2\epsilon(k-1)}{2k(k-1)}}{\frac{1-\epsilon}{(k-1)k}+\frac{\epsilon (k-1)}{k(k-1)} } -\delta'
  & = & \frac{\frac{k-\epsilon k +2\epsilon k- 2\epsilon}{2(k-1)k}
 }{\frac{1-\epsilon+\epsilon k- \epsilon }{(k-1)k}  } -\delta'\\[1ex]
 & = & \frac{k +  (k- 2)\epsilon}{2(1+(k-2)\epsilon)}
  -\delta' \\[1ex]
  & = & \frac{k }{\frac{2k(1+(k-2)\epsilon)}{k +  (k- 2)\epsilon}}
  -\delta' \\[1ex]
  & = & \frac{k }{\frac{2k+2(k-2)\epsilon+(2k-2)(k-2)\epsilon}{k +  (k- 2)\epsilon}}
  -\delta' \\[1ex]
  & = & \frac{k }{2+\frac{(2k-2)(k-2)\epsilon}{k +  (k- 2)\epsilon}}
  -\delta' \\[1ex]
  & = & \frac{k }{2+\frac{2(k-1)(k-2)\epsilon}{k +  (k- 2)\epsilon}}
  -\delta'
   \label{eq:UGCc1a}
\end{eqnarray}   
Finally, we have to verify that the constructed hypergraph $H'$ is indeed $\epsilon$-dense. 
Notice that $H'$ can have at most $(|V'_1|)^{k}=(\frac n k+\frac{\epsilon}{1-\epsilon}\frac{n}{k})^{k}$
edges. Therefore, we obtain  the following:
$$\frac{\left(\frac{\epsilon}{1-\epsilon}\frac{n}{k}\right)\left(\frac n k+\frac{\epsilon}{1-\epsilon}\frac{n}{k}\right)^{k-1}}
{\left(\frac n k+\frac{\epsilon}{1-\epsilon}\frac{n}{k}\right)^{k}}\quad=\quad
\frac{\frac n k \frac{\epsilon}{1-\epsilon}}{\frac n k\left(1+\frac{\epsilon}{1-\epsilon}\right)}\quad=\quad
\frac{\frac{\epsilon}{1-\epsilon}   }{ \frac{1+\epsilon-\epsilon}{1-\epsilon}  }=\epsilon$$ 
Notice that the constructed hypergraph is also $\ell$-wise $\epsilon$-dense balanced.
Hence, we obtain the same inapproximability factor in this case as well.
\end{proof}

Next, we combine the former construction with a conjecture about Unique Games hardness
of the Vertex Cover problem in balanced $k$-partite $k$-uniform hypergraphs. In particular,
we postulate the following:
  
\begin{conjecture}\label{assumt}
Given a balanced $k$-partite $k$-uniform hypergraph $H=(V_1,..,V_k,E(H))$ with  $k\geq 3$, 
let $OPT$ denote an optimal vertex cover of $H$.
For every $\delta>0$, the following is UGC-hard to decide:
$$|V|\left(\frac{1}{k}-\delta\right)\,\leq\,|OPT|\quad\textrm{  or  }\quad|OPT|\,\leq\,|V|
\left(\frac{2}{k^2}+\delta\right)$$
\end{conjecture}

Combining Conjecture~\ref{assumt} with the construction in Theorem~\ref{ugcmain3}, it yields
the following inapproximability result which matches precisely the approximation upper bound achieved
by our approximation algorithm described in Section~\ref{approxopt}: 

\begin{theorem}\label{thm:conj}
For every $\delta>0$, $\epsilon\in (0,1)$, $\ell\in [k-1]$, and $k\geq 3$, there exists no polynomial time approximation
algorithm with an approximation ratio 
$$\frac{ k }{2+(k-2) \epsilon     } -\delta$$ 
for the Vertex Cover problem in  $\ell$-wise $\eps$-dense
$k$-partite $k$-uniform hypergraphs assuming Conjecture~\ref{assumt}.
\end{theorem} 

\begin{proof}
 The UGC-hard decision question from Conjecture~\ref{assumt}
transforms into the following:
$$
n\left(\frac{1}{k}-\delta\right)+\frac{\epsilon}{1-\epsilon}\frac{n}{k}\,\leq\,|OPT'|\quad\textrm{  or  }\quad
|OPT'|\,\leq\,n\left(\frac{2}{k^2}+\delta\right)+\frac{\epsilon}{1-\epsilon}\frac{n}{k} 
$$
Assuming the UGC, this implies the hardness of approximating the Vertex Cover problem
in  $\eps$-dense
$k$-partite $k$-uniform hypergraphs for every $\delta'>0$ to within:
\begin{eqnarray}
\frac{n\left(\frac{1}{k}-\delta\right)+\frac{\epsilon}{1-\epsilon}\frac{n}{k}  }
{  n\left(\frac{2}{k^2}+\delta\right)+\frac{\epsilon}{1-\epsilon}\frac{n}{k}}
& = & \frac{n\left(\frac{1}{k}-\delta\right)(1-\epsilon)+\frac{\epsilon n}{k}}{n\left(\frac{2}{k^2}+\delta\right)(1-\epsilon)+\frac{\epsilon n}{k} } \\[1ex]
 & = & \frac{\frac{n}{k}}{n\left(\frac2{k^2}\right)(1-\epsilon)+\frac{k\epsilon n}{k^2} } -\delta'\\[1ex]
 &  = & \frac{ k }{2(1-\epsilon)+ k \epsilon} -\delta'\\[1ex]
  &  = & \frac{ k }{2+ (k-2) \epsilon    } -\delta'  \label{eq:UGCc1}
\end{eqnarray}   
\end{proof}

\subsection{NP-Hardness}

Recently, Sachdeva and Saket proved in~\cite{SS11} a nearly optimal NP-hardness
of the Vertex Cover problem on balanced $k$-uniform $k$-partite hypergraphs.
More precisely, they obtained the following inapproximability result:
  \begin{theorem}\cite{SS11}\label{npmain2}
Given a balanced $k$-partite $k$-uniform hypergraph $H=(V,E)$ with $k\geq 4$, let $OPT$ denote an optimal vertex cover of $H$.
For every $\delta>0$, the following is NP-hard to decide:
\begin{eqnarray*}
|V|\left(\frac{k}{2(k+1)(2(k+1)+1)}-\delta\right)&\leq&|OPT|\\[1ex]
\textrm{  or  }\\[1ex]
|V|\left(\frac{1}{k(2(k+1)+1)}+\delta\right)&\geq&|OPT|
\end{eqnarray*}   
\end{theorem}
Combining our reduction from Theorem~\ref{ugcmain2} with Theorem~\ref{npmain2}, we prove the 
following inapproximability result under the assumption $P\neq NP$: 

\begin{theorem}
For every $\delta>0$, $\epsilon \in (0,1) $, $\ell\in [k-1]$, and $k\geq 4$, there is no polynomial time approximation
algorithm with an approximation ratio 
$$\frac{   k^2(1-\epsilon)+\epsilon2(k+1)(2(k+1)+1)  }
{  2(k+1)[1-\epsilon+ \epsilon(2(k+1)+1)]  }-\delta $$ 
for the Vertex Cover problem in $\ell$-wise $\eps$-dense
$k$-partite $k$-uniform hypergraphs assuming $P\neq NP$.
\end{theorem} 

\begin{proof}
 The NP-hard decision question from Theorem~\ref{npmain2}
transforms into the following:
\begin{eqnarray*}
n\left(\frac{k}{2(k+1)(2(k+1)+1)}-\delta\right)+\frac{\epsilon}{1-\epsilon}\frac{n}{k} &\leq& |OPT'| \\
 \textrm{  or  } \\
 n\left(\frac{1}{k(2(k+1)+1)}+\delta\right)+\frac{\epsilon}{1-\epsilon}\frac{n}{k}& \geq &  |OPT'|
\end{eqnarray*}   
Assuming $NP\neq P$, this implies the hardness of approximating the Vertex Cover problem
in  $\eps$-dense hypergraphs for every $\delta'>0$ to within:
\begin{eqnarray}
\frac{   \frac{k(1-\epsilon)}{2(k+1)(2(k+1)+1)}+\frac{\epsilon}{k}  }
{  \frac{1-\epsilon}{k(2(k+1)+1)}+\frac{\epsilon}{k}} -\delta'
& = &  \frac{   \frac{k^2(1-\epsilon)+\epsilon2(k+1)(2(k+1)+1)}{k2(k+1)(2(k+1)+1)}  }
{  \frac{1-\epsilon+ \epsilon(2(k+1)+1) }{k(2(k+1)+1)} } -\delta'\\
 & = & \frac{   k^2(1-\epsilon)+\epsilon2(k+1)(2(k+1)+1)  }
{  2(k+1)[1-\epsilon+ \epsilon(2(k+1)+1)]  } -\delta'
\end{eqnarray}   
\end{proof}

\section{Further Research}
An interesting question remains about even tighter lower approximation bounds for
our problem, perhaps connecting it more closely to the integrality gap issue of
the LP of Lov\'asz~\cite{L75}.
  
\section*{Acknowledgment}
We thank Jean Cardinal for many stimulating discussions.

\end{document}